\documentclass[a4paper,USenglish,cleveref, autoref, thm-restate]{lipics-v2021}

\hideLIPIcs  

\usepackage{amsmath,array,bm}
\usepackage{amsfonts}
\usepackage{amsthm}
\usepackage{amssymb,latexsym}
\usepackage{subcaption}
\usepackage{algorithmicx}
\usepackage{marginnote}
\usepackage[]{algpseudocode}
\usepackage[dvipsnames]{xcolor}
\usepackage{graphicx}
\usepackage{caption}
\usepackage{float}
\usepackage{thmtools}
\usepackage{hyperref}
\usepackage{mathtools}
\usepackage{thm-restate}
\usepackage{cleveref}
\usepackage{bussproofs}
\usepackage[square,sort,comma,numbers]{natbib}
\usepackage{tikz}
\usepackage{xspace}
\usepackage{xcolor}
\usepackage{tipa}
\usepackage{mdframed}
\usepackage[export]{adjustbox}[2011/08/13]
\usepackage{multirow}
\usepackage{hhline}
\usepackage{todonotes}
\usepackage{tcolorbox} 
\usepackage{dirtytalk}

\usepackage{tabularx}
\usepackage{nicefrac}
\usepackage{epsfig}
\usepackage{xfrac}

\newcommand{\hidestuff}[1]{}

\usepackage{tcolorbox}
\usepackage{tikz}

\title{Hybrid $k$-Clustering: Blending $k$-Median and $k$-Center}

\titlerunning{Hybrid $k$-Clustering: Blending $k$-Median and $k$-Center}

\author{Fedor V. Fomin}{University of Bergen, Norway}{Fedor.Fomin@uib.no}{https://orcid.org/0000-0003-1955-4612}{Supported by the Research Council of Norway via the project  BWCA (grant no. 314528).}

\author{Petr A. Golovach}{University of Bergen, Norway}{Petr.Golovach@ii.uib.no}{https://orcid.org/0000-0002-2619-2990}{Supported by the Research Council of Norway via the project  BWCA (grant no. 314528).}

\author{Tanmay Inamdar}
{Indian Institute of Technology Jodhpur, Jodhpur, India}{taninamdar@gmail.com}{https://orcid.org/0000-0002-0184-5932}{}

\author{Saket Saurabh}{The Institute of Mathematical Sciences, HBNI, Chennai, India  \and University of Bergen, Norway }{saket@imsc.res.in}{https://orcid.org/0000-0001-7847-6402}{The author is supported by the European Research Council (ERC) under the European Union's Horizon 2020 research and innovation programme (grant agreement No. 819416); and he also acknowledges the support of Swarnajayanti Fellowship grant DST/SJF/MSA-01/2017-18.}

\author{Meirav Zehavi}{Ben-Gurion University of the Negev, Beer-Sheva, Israel}{meiravze@bgu.ac.il}{https://orcid.org/0000-0002-3636-5322}{The research was supported by the European Research Council (ERC) grant no. 101039913 (PARAPATH).}

\Copyright{Fedor V. Fomin, Petr A. Golovach, Tanmay Inamdar, Saket Saurabh, Meirav Zehavi}

\authorrunning{F. V. Fomin, P. A. Golovach, T. Inamdar, S. Saurabh, M. Zehavi}

\ccsdesc[500]{Theory of computation~Facility location and clustering}
\ccsdesc[500]{Theory of computation~Fixed parameter tractability}

\keywords{clustering, $k$-center, $k$-median, Euclidean space, fpt approximation} 

\nolinenumbers 

\EventEditors{Amit Kumar and Noga Ron-Zewi}
\EventNoEds{2}
\EventLongTitle{Approximation, Randomization, and Combinatorial Optimization. Algorithms and Techniques (APPROX/RANDOM 2024)}
\EventShortTitle{\mbox{\scriptsize APPROX/RANDOM 2024}}
\EventAcronym{APPROX/RANDOM}
\EventYear{2024}
\EventDate{August 28--30, 2024}
\EventLocation{London School of Economics, London, UK}
\EventLogo{}
\SeriesVolume{317}
\ArticleNo{4}

\usepackage{amsmath,array,bm}
\usepackage{amsmath,amsfonts}
\usepackage{amsthm}
\usepackage{soul}
\usepackage{amssymb,latexsym}
\usepackage[ruled]{algorithm}
\usepackage{subcaption}
\usepackage{algorithmicx}
\usepackage[]{algpseudocode}
\usepackage{bbm}
\usepackage{mathrsfs}
\usepackage{graphicx}
\usepackage{caption}
\usepackage{eufrak}
\usepackage{float}
\usepackage{thmtools}
\usepackage{hyperref}
\usepackage{mathtools}
\usepackage{cleveref}
\usepackage{bussproofs}
\usepackage[square,sort,comma,numbers]{natbib}
\usepackage{tikz}
\usepackage{stmaryrd}
\usepackage{xspace}
\usepackage{xcolor}
\usepackage{marginnote}
\usepackage{tipa}
\usepackage{mdframed}
\usepackage[export]{adjustbox}[2011/08/13]
\usepackage{multirow}
\usepackage{relsize}
\usepackage{hhline}
\usepackage{todonotes}
 \usepackage{tcolorbox} 
\usepackage{nicefrac}

\crefname{invar}{invariant}{invariants}
\crefname{ineq}{inequality}{inequalities}
\crefname{constr}{constraint}{constraints}
\crefname{tbl}{table}{tables}
\crefname{lem}{lemma}{lemmata}
\crefname{lemma}{lemma}{lemmata}
\crefname{cond}{condition}{conditions}

\newcommand{\red}[1]{{\color{red} #1}}
\newcommand{\blue}[1]{{\color{blue} #1}}
\newcommand{\lr}[1]{\left( #1\right)}
\newcommand{\LR}[1]{\left\{ #1\right\}}

\newcommand{\Oh}{\mathcal{O}}

\newcommand{\cI}{\mathcal{I}}

\newcommand{\OPT}{{\sf OPT}}

\renewcommand{\d}{\mathsf{dist}}
\newcommand{\fac}{\mathbb{F}}
\newcommand{\cli}{P}
\newcommand{\cen}{F}
\newcommand{\real}{\mathbb{R}}
\newcommand{\cost}{\mathsf{cost}}
\newcommand{\opt}{{\sf OPT}}
\newcommand{\grid}{{\sf Grid}}

\newcommand{\cluster}{\mathsf{cl}}
\newcommand{\new}{\text{new}}

\newcommand{\Nnear}{N_{\text{near}}}
\newcommand{\Nfar}{N_{\text{far}}}

\newcommand{\dmax}{\d_{\max}}
\newcommand{\Lnear}{L_{\text{near}}}
\newcommand{\Lfar}{L_{\text{far}}}

\newcommand{\kmed}{$k$-\textsc{Median}\xspace}
\newcommand{\kcenter}{$k$-\textsc{Center}\xspace}
\newcommand{\probname}{{\sc Hybrid $k$-Clustering}\xspace}

\newtcolorbox{mybox}[1]{colback=white!5!white,colframe=gray!75!black,colbacktitle=white!5!white,coltitle=black!70!black,sharp corners=all,title={#1}}

\begin{document}
\maketitle

\begin{abstract}
We propose a novel clustering model encompassing two well-known clustering models: $k$-center clustering and $k$-median clustering. In the \probname problem, given a set $P$ of points in $\Bbb{R}^d$, an integer $k$, and a non-negative real $r$, our objective is to position $k$ closed balls of radius $r$ to minimize the sum of distances from points not covered by the balls to their closest balls. Equivalently, we seek an optimal $L_1$-fitting of a union of $k$ balls of radius $r$ to a set of points in the  Euclidean space. When $r=0$, this corresponds to \textsc{$k$-median}; when the minimum sum is zero, indicating complete coverage of all points, it is \textsc{$k$-center}.

Our primary result is a bicriteria approximation algorithm that, for a given $\varepsilon>0$, produces a hybrid $k$-clustering with balls of radius $(1+\varepsilon)r$. This algorithm achieves a cost at most $1+\varepsilon$ of the optimum, and it operates in time $2^{(kd/\varepsilon)^{\Oh(1)}} \cdot n^{\Oh(1)}$. Notably, considering the established lower bounds on $k$-center and $k$-median, our bicriteria approximation stands as the best possible result for \probname.
\end{abstract}

\section{Introduction}

Suppose we want to install a set of $k$ access points (APs) at certain locations to provide wireless internet (Wi-Fi) coverage to a group of people belonging to a certain area. Each AP is capable of providing Wi-Fi within a circular-shaped region (i.e., a \emph{disk}) of fixed radius $r$, and it may not be possible to cover the entire region with $k$ such disks. Thus, after placing $k$ APs, some people may be \emph{outliers}, that lie outside any of the $k$ disks and do not receive Wi-Fi coverage. We can model this scenario as the classical {\sc $k$-Center with Outliers} problem, which is a crude model since it only cares about the \emph{number} of outliers. However, our scenario is more nuanced. All people that lie within any of the $k$ disks of radius $r$ already receive Wi-Fi, whereas a person lying outside all of the $k$ disks must travel to the boundary of the nearest disk in order to receive coverage. Naturally, we would like to minimize the total distance traveled by people. Motivated by this and several other problems in computational geometry/clustering, we consider the following clustering problem, which encompasses two fundamental variants of clustering: \kcenter and \kmed. Given a set $P$ of points in some metric space and integer $k$ and real $r\geq 0$, our objective is to position $k$ closed balls of radius $r$ in a way that minimizes the sum of distances from points uncovered by the balls to their closest balls. In \Cref{fig-example}, we provide an example of such clustering with $k=2$ and $r=2$.

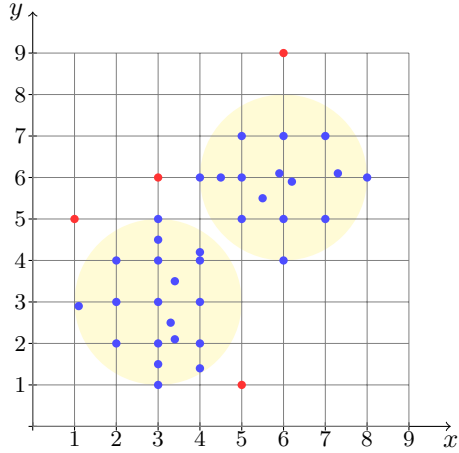
\begin{figure}
	\centering
\begin{tikzpicture}[scale=0.55]
      \fill[yellow!20!white] (3,3) circle (2cm);
        \fill[yellow!20!white] (6,6) circle (2cm);
  \draw[->] (0,0) -- (10,0) node[below] {$x$};
  \draw[->] (0,0) -- (0,10) node[left] {$y$};

  \foreach \x in {0,1,2,3,4,5,6,7,8,9}
     \draw (\x,-0.1) -- (\x,0.1);
   \foreach \y in {0,1,2,3,4,5,6,7,8,9}
    \draw (-0.1,\y) -- (0.1,\y);

  \foreach \x in {1,2,3,4,5,6,7,8,9}
    \node at (\x,-0.3) { {\footnotesize $\x$}};
  \foreach \y in {1,2,3,4,5,6,7,8,9}
    \node at (-0.3,\y) { {\footnotesize $\y$}};

    \draw[step=1cm,gray,very thin] (0,0) grid (9,9);

 \foreach \x in {2,3,4} \fill[blue!70!white] (2,\x) circle (0.1cm) ;
       \foreach \x in {1.4,2,3,4,4.2}    \fill[blue!70!white] (4,\x) circle (0.1cm) ;
    \foreach \x in {1,1.5,2,3,4,4.5,5}    \fill[blue!70!white] (3,\x) circle (0.1cm) ;  
    
   \foreach \x in {5,6,7} \fill[blue!70!white] (5,\x) circle (0.1cm) ;
       \foreach \x in {4,5,7}    \fill[blue!70!white] (6,\x) circle (0.1cm) ;
    \foreach \x in {5,7}    \fill[blue!70!white] (7,\x) circle (0.1cm) ;  
     \fill[blue!70!white] (8,6) circle (0.1cm) ;  
       \fill[blue!70!white] (4,6) circle (0.1cm) ;  
         \fill[blue!70!white] (4.5,6) circle (0.1cm) ;   
           \fill[blue!70!white] (5.5,5.5) circle (0.1cm) ;   
             \fill[blue!70!white] (7.3,6.1) circle (0.1cm) ;   
                     \fill[blue!70!white] (3.3,2.5) circle (0.1cm) ;       
                      \fill[blue!70!white] (3.4,2.1) circle (0.1cm) ;  
                      \fill[blue!70!white] (1.1,2.9) circle (0.1cm) ;  
                       \fill[blue!70!white] (3.4,3.5) circle (0.1cm) ;  
                               \fill[blue!70!white] (5.9,6.1) circle (0.1cm) ;  
                                 \fill[blue!70!white] (6.2,5.9) circle (0.1cm) ; 
  \fill[red!80!white] (3,6) circle (0.1cm) ;        
   \fill[red!80!white] (1,5) circle (0.1cm) ;      
      \fill[red!80!white] (5,1) circle (0.1cm) ;              
              \fill[red!80!white] (6,9) circle (0.1cm) ;         
    
\end{tikzpicture} 
\caption{\small Two disks of radius $2$ cover all except four points that are colored red. The total sum of distances from these points to the yellow disks is $2(1+ \sqrt{8}-2)$.}\label{fig-example}
\end{figure}

To define the new clustering formally, we need some definitions. We consider Euclidean inputs, i.e., all points belong to $\real^d$ for some $d \ge 1$ and the distance function $\d(\cdot, \cdot)$ is given by the Euclidean ($\ell_2$) distance. For a point $p \in P$ and a finite set of points $Q \subset \real^d$, we define $\d(p, Q) \coloneqq \min_{q \in Q} \d(p, q)$. Further, for $x, y \in P$, and a real $r \ge 0$, we define the shorthand $\d_r(x, y) \coloneqq \max\LR{\d(x, y)-r, 0}$.  

\begin{mybox}{\probname}
	\textbf{Input.} A set $\cli \subset \mathbb{R}^d$ of $n$ points, an integer $k \ge 1$, and a real $r \ge 0$.
	\\\textbf{Task.} Find a set $\cen \subset \mathbb{R}^d$ of size at most $k$, that minimizes:
	\begin{equation}
		\cost_r(\cli, \cen) \coloneqq \sum_{p \in \cli} \d_r(p, \cen) \label{eqn:objective}
	\end{equation}
\end{mybox}
We denote an instance of \probname as $\cI = (\cli, k, r, d)$, where $d$ denotes the dimension. When $r = 0$, the optimal cost of \probname equals the optimal \kmed clustering cost of the instance. Thus in this case, 
\probname reduces to \kmed.  
However,  when $r > 0$,  $\d_r(\cdot, \cdot)$ does not form a metric, and hence we cannot simply reduce the problem to \kmed.
On the other hand, the minimum value $r$ that guarantees the cost of \probname to be zero  
is equal to the optimal \kcenter value. In this sense, \probname reduces to  \kcenter. 

\begin{figure}[ht]
	\centering
	\includegraphics[scale=0.85,page=5]{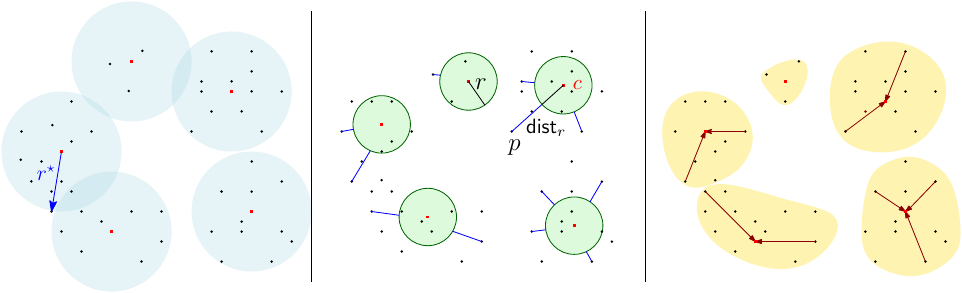}
	\caption{\small Left: \kcenter clustering, a special case of \probname with $r = r^\star$. All points are covered by $k$ balls of radius $r^\star$ and $\opt_{r^\star} = 0$.
	Right: \kmed clustering, a special case of \probname with $r = 0$, and every point contributes its distance to the closest center (some are shown as brown arrows). Middle: A general instance of \probname lies somewhere \emph{in between} the two cases, where points outside radius-$r$ balls contribute the distance to the boundary (shown in blue).}
\end{figure}

\subsection{Our Result and Techniques}
The main result of this paper is a  bicriteria approximation algorithm for \probname.
An $\alpha$-approximation to an instance $\cI = (\cli, k, r, d)$ is a subset $\cen \subset \real^d$ of size $k$ with $\cost_r(\cli, \cen) \le \alpha \cdot \opt_r$, where $\opt_r \coloneqq \cost_r(\cli, \cen^*)$ denotes the cost of an optimal solution $\cen^* \subset \real^d$ of size at most $k$. Furthermore, an $(\alpha, \beta)$-bicriteria approximation is a solution $\cen \subset \real^d$ with $\cost_{\beta r}(\cli, \cen) \le \alpha \cdot \opt_r$. Here, $\cost_{\beta r}(\cli, \cen) = \sum_{p \in \cli} \d_{\beta r}(\cli, \cen)$.

Consider the special case of $r = r^*$, where $r^*$ is the optimal radius for \kcenter. Then, $\opt_{r^*} = 0$. Therefore, a $(\alpha, 1)$-bicriteria approximation would return a solution of cost $\alpha \cdot \opt_{r^*} = 0$ using radius $1 \cdot r^*$, i.e., an optimal solution for \kcenter. On the other hand, a $(1, \beta)$-bicriteria approximation, for the special case of $r = 0$, would return an optimal-cost solution using the radius of $\beta r = 0$. That is, such an algorithm would optimally solve \kmed. Combining these observations with the established lower bounds from the literature for \kmed and \kcenter in Euclidean spaces, implies the following bounds for \probname.
\begin{proposition} \label{obs:lowerbounds}\ 
	The following holds for \probname even when the input is from $\real^2$.
	\begin{itemize}
		\item For any $\alpha \ge 1$, there exists no $FPT$ in  $k$ algorithm that returns an $(\alpha, 1)$-approximation, unless FPT $= \mathsf{W}[1]$ \cite{Marx05}.
		\item For any finite $\beta \ge 1$, there exists no polynomial-time algorithm that returns a $(1, \beta)$-approximation unless $\mathsf{P = NP}$ \cite{megiddo1984complexity}. 
	\end{itemize}
	Further, assuming the Exponential-Time Hypothesis (ETH), if the input is from $\real^d$ with $d \ge 4$, then there exists no $n^{o(k)}$ time algorithm that returns a $(1, \beta)$-approximation, for any finite $\beta \ge 1$ \cite{CASODA18}. 
\end{proposition}

Given these results, a natural question arises: \emph{Can we achieve a $(1+\varepsilon, 1+\varepsilon)$-approximation for \probname, running in time $f(k, \varepsilon) \cdot n^{\Oh(1)}$, particularly in low-dimensional Euclidean spaces?}
Our main theorem answers this question. 

\begin{restatable}{theorem}{fptasthm} \label{thm:fptas}
	Let $0 < \varepsilon < 1$. There exists a randomized algorithm that, given an instance of \probname in $\real^d$, runs in time $2^{\lr{\frac{kd}{\varepsilon}}^{\Oh(1)}} \cdot n^{\Oh(1)}$, and returns a $(1+\varepsilon, 1+\varepsilon)$-approximation with probability at least a positive constant.
\end{restatable}

This randomized algorithm and the proof of correctness are described in \Cref{sec:euclidean}. Here we discuss some of the main ideas. Recall that our objective, as the problem name suggests, is a ``hybrid'' of \kcenter and \kmed. In our preprocessing steps, we first handle the inputs that behave \emph{almost like} either of the two problems. Suppose we (approximately) know the optimal value of \probname for the given set of points $P$, called $\opt_r$. First, in \Cref{lem:kcenterlemma}, if $r > \opt_r$, then we show that an approximate solution can be found using techniques used for approximating \kcenter. Specifically, for each of the $k$ centers in the optimal solution, we find a ``nearby'' center within distance $\epsilon r$ via overlaying a fine grid in the space. Thus, we can assume that $r \le \opt_r$. Next, we consider the case when $r$ is \emph{too small} compared to $\opt_r$, namely, when $r < \frac{\varepsilon \opt_r}{n}$, and show that in this case, the input behaves like \kmed -- an approximate \kmed solution is also an approximation for \probname (\Cref{lem:kmedlemma}). In this manner, we preprocess to handle inputs that resemble \kcenter and \kmed, we obtain a relation between $r$ and $\opt_r$, which can be used to discretize the distances, which can be used to bound the aspect ratio (i.e., the ratio of maximum to minimum positive distance) (\Cref{lem:aspectratio}).

After the preprocessing step, we obtain inputs that are not \emph{immediately reducible} to \kcenter/\textsc{Median}. To handle such inputs, we design an intricate recursive algorithm that, at each step, tries to simultaneously handle parts (i.e., clusters) of the input that can be handled by either of the two techniques. This algorithm is inspired by the sampling approach of Kumar, Sabharwal, and Sen~\cite{KumarSS05,KumarSS10} (also Jaiswal, Kumar, and Sen~\cite{JKS2014}). In this approach, one first takes a large enough sample that can be used to pin down the location of the largest cluster center. Then, one removes enough points from the vicinity of this center, so that the next largest cluster becomes dominant, and hence a subsequent sample contains sufficiently many points from the second cluster, and so on.

\begin{figure}[t]
	\centering
	\includegraphics[scale=0.6,page=1]{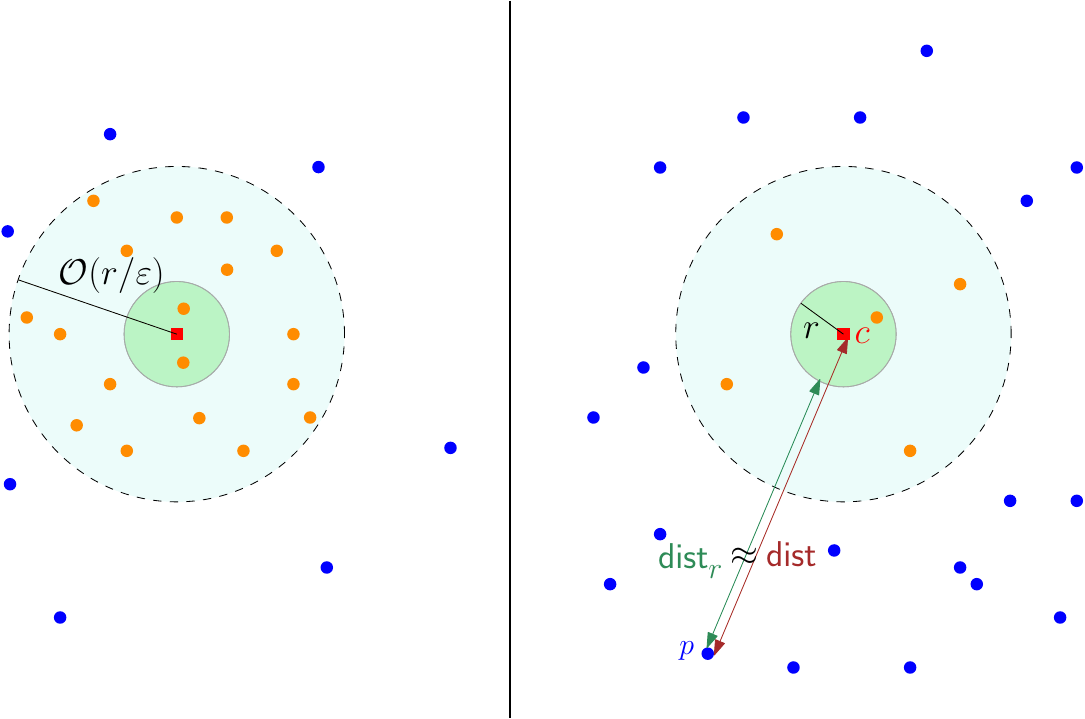}
	\caption{\small Example of two different types of clusters. In each figure, we show the cluster center in \red{red}, a ball of radius $r$ around the center in green, and a larger ball of radius $\Oh(r/\epsilon)$ in cyan with a dashed outline. Left: A \emph{$1$-center-like} cluster. Note that a large chunk of points lies within the radius $\Oh(r/\varepsilon)$ ball around the center. Right: A \emph{$1$-median-like} cluster. Note that most of the points lie outside the $\Oh(r/\epsilon)$ radius ball around $\red{c}$, and for any such point, e.g., \blue{$p$} that is outside the $\Oh(r/\varepsilon)$ radius ball, $\d_r(\blue{p}, \red{c}) \approx \d(\blue{p}, \red{c})$.} \label{fig:introclusters}
\end{figure} 

However, our scenario is more intricate and challenging for several reasons due to the peculiar nature of the objective. Nevertheless, in principle, one can classify each cluster as either being more \emph{$1$-center-like}, or more \emph{$1$-median like} (see \Cref{fig:introclusters} for an illustration). In a \emph{$1$-center-like} cluster, a large fraction of points lie within a ball of radius $\Oh(r/\varepsilon)$. On the other hand, in a \emph{$1$-median-like} cluster, a vast majority of points lie outside the $\Oh(r/\varepsilon)$-radius ball. Note that any such point loses very little due to the ``$-r$'' term in the clustering cost, i.e., its $\d_r$ and $\d$ values are approximately equal.  Hypothetically, if we knew the partition of the input points into $k$ clusters, then we could use this classification to handle each type of cluster separately -- an almost-optimal center of a \emph{$1$-center-like} cluster can be found using a grid, whereas one can use an approximation for $1$-median (as a black box) to handle a \emph{$1$-median-like} cluster. However, the actual clusters are obviously unknown to the algorithm. Hence, the algorithm has to carefully navigate between the two types of clusters based on the random sample obtained, and must simultaneously handle both scenarios using branching (i.e., recursion). The analysis of the algorithm is also much more involved due to the various cases in which the distinction between two types of clusters is murkier. Nevertheless, we are able to show that the algorithm returns a $(1+\varepsilon, 1+\varepsilon)$-approximation in time $2^{(kd/\varepsilon)^{\Oh(1)}} \cdot n^{\Oh(1)}$ with good probability. Note that we incur an exponential dependence on the dimension $d$ due to ``grid-arguments'' used to handle \emph{$1$-center-like} clusters, unlike the approach of \cite{KumarSS05}. However, such dependence seems unavoidable using our approach.

\subsection{Related Problems}

\subparagraph{Euclidean Clustering.}
An extensive body of literature exists on approximation algorithms for \kcenter and \kmed in the Euclidean space. For \kmed in $\real^d$, Polynomial-Time Approximation Schemes (PTASes) with a running time of $n^{f(\epsilon, d)}$ have been developed, leveraging local search techniques \cite{DBLP:journals/siamcomp/FriggstadRS19, DBLP:journals/siamcomp/Cohen-AddadKM19}. Additionally, various Fixed-Parameter Tractable Approximation Schemes (FPT-AS) with a running time of $f(k, d, \epsilon) \cdot n^{\Oh(1)}$ are known for this problem \cite{Cohen-AddadG0LL19Tight, Cohen-AddadLSS22, KumarSS10,JKS2014}. The dependence on dimension $d$ can be eliminated through dimensionality reduction techniques \cite{MakarychevMR19, Charikar22JL}.

For \kcenter, an FPT-AS was introduced by Agarwal and Procopiuc in \cite{AgarwalP02}, with a runtime of $\Oh(n \log k) + (k/\varepsilon)^{\Oh(dk^{1-1/d})}$ in $\real^d$. Subsequent work by Badoiu, Har{-}Peled, and Indyk \cite{BadoiuHI02} improved the running time to $2^{\Oh(k \log k)/\epsilon^2}$.

In \cite{Tamir01}, Tamir introduces a common generalization of the two clustering problems, namely, \textsc{$\ell$-centrum}. In this problem, one ignores $\ell$ closest points from the cost. Notably, \kmed ignores $0$ points, and \kcenter ignores all but one point. While this problem is related to \probname, their objectives differ. \textsc{Ordered} \kmed, a further generalization of $\ell$-\textsc{Centrum}, also does not align with our objective. Approximation algorithms for this problem and some variants were developed in \cite{AbbasiClustering23, AbbasiRobustClustering23,ByrkaSS18,ChakrabartyS18}.

\subparagraph{$k$-center clustering with outliers.}
In  \kcenter clustering, we are given sets $\cli$  (clients) and $\fac$ (facilities)  of points. Given an integer $k$, the task is to identify $k$ centers $F \subseteq \fac$ minimizing the maximum distance of any point in $\cli$ from its closest center. A popular variant of $k$-center is a formulation that considers outliers. For a selected parameter $x$, up to $x$ points are allowed not to be allocated to any center. A plethora of approximation algorithms for this problem, and the related problems of covering points by disks and minimum enclosing balls with outliers, exist in the literature \cite{BadoiuHI02,BergBM23,CharikarKMN01,gandhi2004approximation,DingYW19,EfratSZ94,Har-PeledM05,Matousek95}. \probname could be seen as a variant of \kcenter with outliers, where
we focus on the sum of distances to outliers rather than their numbers. 

\subparagraph{Shape fitting.} A natural problem arising in machine learning, statistics, data-mining, and many other fields is to fit a shape $\gamma$ to a set of points $P$ in $\Bbb{R}^d$. Har-Peled in~\cite{Har-Peled07} introduces the following formalization of this problem. For a family of shapes $\mathcal{F}$ (points, lines, hyperplanes,  spheres, etc.) we seek for a shape $\gamma\in \mathcal{F}$ with the best fit to $P$. 
The typical criteria for measuring how well a shape $\gamma$ fits a set of points $P$ could be 
the maximum distance between a point of $P$ and its nearest point on $\gamma$ ($L_\infty$-fitting), sum of the distances from $P$ to $\gamma$ ($L_1$-fitting) or the sum of the squares of the distances ($L_2$-fitting). In this setting, \probname is the problem of $L_1$-fitting to a shape from 
$\mathcal{F}$, where $\mathcal{F}$ is the family of shapes defined by unions (not necessarily disjoint) of $k$ balls in $\Bbb{R}^d$. Some relevant work in this direction includes  \cite{AgarwalHY08,Har-Peled07,Har-PeledW04,YuAPV08}.

\section{Bicriteria FPT Approximation Scheme in Euclidean Spaces} \label{sec:euclidean}

We first set up some notation and define an important subroutine. For $Y \subset \real^d$, and $y \in Y$, let $\cluster(y, Y) \subseteq \cli$ denote the subset of points of $\cli$, whose closest point in $Y$ is $y$. Ties are broken arbitrarily. Note that $\LR{\cluster(y, Y): y \in Y}$ forms a partition of $\cli$. 

Let $p \in \real^d$ be a point and $\lambda \ge 0$, let $B(p, \lambda) = \LR{ q \in \real^d : d(p, q) \le \lambda}$ denote the ball of radius $\lambda$ centered at $p$. For $0\le \tau \le \lambda$, let $\grid(p, \lambda, \tau) $ be the outcome of the following procedure: we place a grid of sidelength $\tau/\sqrt{d}$ (of arbitrary offset). From each grid cell $L$ that (partially) intersects with $B(p, \lambda)$ (i.e., $L$ contains a point $q$ with $d(p, q) \le \lambda$),  we pick an arbitrary point from $L$ and add it to the set $\grid(p, \lambda, \tau)$. Note that $\grid(p, \lambda, \tau)$ can be computed in time proportional to the size of the output. We have the following observations that follow from simple geometric arguments.

\begin{observation} \label{obs:grid}~
	\begin{enumerate}
		\item $|\grid(p, \lambda, \tau)| \le \Oh((\sqrt{d}\lambda/\tau)^d)$, where $d$ is the dimension.
		\item For any $q \in B(p, \lambda)$, there exists some $q' \in \grid(p, \lambda, \tau)$ such that $d(q, q') \le \tau$.
	\end{enumerate}
\end{observation}

\subsection{Preprocessing}

Suppose we know an estimate of $\OPT_r$ up to a constant factor -- this can be done by an exponential search or by first finding a bicriteria (constant) approximation. For simplicity of exposition, we assume that we know $\OPT_r$ exactly.

\subparagraph{Step 1.} Obtaining $\opt_r \ge r \ge \frac{\varepsilon \opt_r}{2n}$.

First, in the following lemma, we handle \emph{$k$-center-like} instances, which we can handle using ``grid arguments''. If this is not applicable, we obtain that $r \le \opt_r$.

\begin{lemma}\label{lem:kcenterlemma}
	If $r > \OPT_r$, then in time $\lr{\frac{d}{\varepsilon}}^{\Oh(dk)} \cdot n^{\Oh(1)} $one can find a set $\cen \subset \real^d$ of size $k$, such that, $\cost_{(1+\varepsilon)r}(\cli, \cen) \le (1+\varepsilon) \OPT_r$.
\end{lemma}

\begin{proof}
	Consider an optimal solution $\cen^\star$ such that $\cost_r(\cli, \cen^\star) = \OPT_r$. We assume that $\cen^\star$ is \emph{minimal}, i.e., no strict subset of $\cen^\star$ also yields the optimal cost. It follows that every point $p \in \cli$ is within a distance of at most $2r$ from some $c \in \cen^\star$, and each $c \in \cen^\star$ has some $p \in \cli$ within distance $2r$. Thus, $\cli$ can be covered using at most $k$ balls of radius $2r$. Using a $2$-approximation for \textsc{$k$-Center}, we find a set $\cen$ such that $\cli \subseteq \bigcup_{c \in \cen}B(c, 4r)$. This implies that $\cen^\star \subseteq \bigcup_{c \in \cen}B(c, 6r)$.
	
	Let $R = \bigcup_{c \in \cen} \grid(c, 6r, \varepsilon r)$.  By \Cref{obs:grid}, $|R| \le k\lr{\frac{\sqrt{d}}{\varepsilon}}^{\Oh(d)}$. Now, we iterate over all subsets $\cen' \subseteq R$ of size at most $k$, and look at the subset $\cen'$ minimizing $\cost_{r}(\cli, \cen')$. The number of such subsets is at most 
	$k\binom{|R|}{k} \le k\lr{\frac{e|R|}{k}}^k \le \lr{\frac{d}{\varepsilon}}^{\Oh(dk)}$.
	Also, by \Cref{obs:grid}, there exists a subset $\cen \subseteq R$ of size at most $k$, such that for each $c^\star \in \cen^\star$, there exists some $c' \in \cen$ such that $\d(c^\star, c') \le \varepsilon r$. It is easy to see that $\cen$ satisfies the claimed properties.
\end{proof}

In the following lemma, we handle \emph{$k$-median-like} instances, where $r$ is very small compared $\OPT_r$. We directly reduce such instances to \kmed (where $r = 0$). If this is not applicable, then we obtain that $r$ is not ``too small'' compared to $\OPT_r$.

\begin{lemma}\label{lem:kmedlemma}
	Let $0 < \varepsilon < 1$. 
	Suppose for an instance $\cI$, $\opt_r \ge \frac{2nr}{\varepsilon}$. Then, $\opt_0 \le (1+\varepsilon/2) \cdot \opt_r \le (1+\varepsilon/2) \cdot \opt_0$. Furthermore, if $\cen \subset \real^d$ satisfies that $\cost_0(\cli, \cen) \le (1+\varepsilon/3) \cdot \opt_0$. Then, $\cost_r(\cli, \cen) \le (1+\varepsilon) \cdot \opt_r$. Such a set $\cen$ can be found in time $2^{(k/\varepsilon)^{\Oh(1)}} \cdot nd$.
\end{lemma}

\begin{proof}
	For any $\cen' \subset \real^d$, it holds that $\cost_0(\cli, \cen) \le \cost_r(\cli, \cen') + nr$. In particular, let $\cen^* \subset \real^d$ be an optimal solution of size $k$, i.e., $\cost_r(\cli, \cen^*) = \opt_r$. Then, it holds that $\cost_0(\cli, \cen^*) \le \cost_r(\cli, \cen^*) + n \cdot \frac{\varepsilon \opt_r}{2n} = (1+\varepsilon/2) \cdot \opt_r$. Since $\cost_0(\cli, \cen^*) \ge \opt_0$, since $\cen^*$ is a feasible solution of size $k$ for $r = 0$. This shows the first inequality.
	
	Now, let $\cen \subset \real^d$ be a $(1+\varepsilon/3)$-approximate solution of size $k$ for $\cost_0$, as in the statement of the lemma, i.e., $\cost_0(\cli, \cen) \le (1+\varepsilon/3) \cdot \opt_0$. Then, the first inequality implies that $\cost_0(\cli, \cen) \le (1+\varepsilon/3) \cdot (1+\varepsilon/2) \cdot \opt_r \le (1+\varepsilon) \cdot \opt_r$. We can use a $(1+\varepsilon/3)$-approximation algorithm (e.g., \cite{KumarSS05,KumarSS10}) for \kmed to find such a solution.
\end{proof}

For a given input $P$, we try the procedures from \Cref{lem:kcenterlemma} and \ref{lem:kmedlemma} and keep them as candidate solutions. However, if $P$ does not satisfy the conditions required to apply these lemmas, then we must have that $\frac{\varepsilon \opt_r}{2n} \le r \le \opt_r$. In this case, we use the next step before proceeding to the main algorithm.

\subparagraph{Step 2.} Bounding the aspect ratio.

In this step, we suitably discretize the distances in order to bound the aspect ratio of the metric
(i.e., the maximum ratio of inter-point distances) by $\Oh(\frac{n^2}{\varepsilon})$. This procedure preserves the cost of an optimal solution up to a factor of $1+\varepsilon$.

\begin{lemma}\label{lem:aspectratio}
	Let $\cli$ be a set of points satisfying $\frac{\varepsilon \opt_r}{2n} \le r \le \opt_r$. Then, in polynomial time we can obtain another (multi)set of points $\cli'$ such that, for any solution $\cen \subset \real^d$, $\cost_r(\cli', \cen) \in (1\pm \varepsilon) \cdot \cost_r(\cli, \cen)$, and $\displaystyle \frac{\max_{p, q \in \cli'} \d(p, q)}{\min_{p, q \in P': \d(p, q) \neq 0} \d(p, q)} \le \frac{4n^2}{\varepsilon}$.
\end{lemma}

\begin{proof}
	Imagine an auxiliary graph $G = (\cli, E)$, where $pq \in E$ iff $\d(p, q) \le 2(\OPT_r+r) \le 4\OPT_r$. If $G$ has more than one connected component, then, note that two points belonging to different connected components cannot belong to the same optimal cluster. Hence, we can solve the problem separately on the points belonging to different connected components and combine the solutions by a simple dynamic programming.
	
	Thus, we can handle each connected component of $G$ separately. In any connected component, the maximum distance between any two points is at most $n \cdot 4\OPT_r$. Now, we place a grid of sidelength $\frac{\varepsilon \OPT_r}{\sqrt{d} n}$ and move each point $p \in \cli$ to the center of the grid. It can be easily shown that this process does not change the cost of an optimal solution by more than a $(1+\varepsilon)$ factor. Thus, the smallest non-zero inter-point distance is now at least $\frac{\varepsilon \OPT_r}{n}$. Thus, the aspect ratio of the instance is bounded by $\frac{4n^2}{\varepsilon}$. 
\end{proof}

Bounding the aspect ratio by $\Oh(\frac{n^2}{\varepsilon})$ means that an exponential search over distances has at most $\log_{2}\lr{\frac{n^2}{\varepsilon}} = \Oh\lr{\frac{\log(n)}{\varepsilon}}$ levels, which will be useful in our main algorithm. By slightly abusing the notation, we continue to use $\cli$ for referring to the discretized (multi)set $\cli'$ returned by \Cref{lem:aspectratio}. If there are any co-located points in $\cli$, we will treat them as separate points, and hence use set terminology instead of multiset terminology.

After the two preprocessing steps, we now proceed to the description of the main algorithm.

\subsection{Main Algorithm}

Our goal is to prove \Cref{thm:fptas}, that is, to  design a randomized bicriteria FPT approximation for \probname.  
We define some parameters. Let $\delta \coloneqq \tfrac{\varepsilon}{10k} < \tfrac{1}{2}$, $\delta' \coloneqq \tfrac{\delta}{3}$ and $r' \coloneqq (1+\delta')r$. 

\begin{algorithm}[h]
	\caption{\texttt{HybridClustering}$(\cen', k, m)$} \label{algo:recursive}
	\begin{algorithmic}[1]
		\Statex $\cen' \subseteq \real^d$ is a subset of centers of size at most $k - m$ added to the solution so far 
		\Statex $\beta = \frac{1}{\delta^{c'}}$ as required in \Cref{prop:samplemedian} and $\beta' \coloneqq \beta \cdot \frac{150k}{\delta^3}$.
		\If{$m = 0$} 
		\State \Return $\cen'$ \label{lin:basecase}
		\EndIf
		\State $R \gets \bigcup_{c' \in \cen'} \grid(c', 16r, \delta r)$ \label{lin:centergrid}
		\For{each $q$ of the form $2^j$ in the range $[8r, \dmax]$} \label{lin:forloop}
		\State $\cli_q \coloneqq \cli \setminus \lr{\bigcup_{c' \in \cen'} B(c', q)}$ \label{lin:faraway}
		\State Let $S_q$ be a sample of size $\beta'$ chosen uniformly at random from $\cli_q$  \label{lin:sq-sample}
		\State $R \gets R \cup \bigcup_{p \in S_q} \grid(p, \frac{8r}{\delta}, \delta r)$ \label{lin:samplegrid}
		\For{each $S \subseteq S_q$ of size $\beta$} \label{lin:innerforloop}
		\State $c' \gets \texttt{ApproxSolutionOnSample}(S, \delta/8)$ \label{lin:approxsample} \Comment{Algorithm from \Cref{prop:samplemedian}}
		\State $R \gets R \cup \LR{c'}$ 
		\EndFor \label{lin:endinnerfor}
		\EndFor \label{lin:endouterfor}
		\For{each $c \in R \setminus \cen'$}
		\State  Call \texttt{HybridClustering}$(\cen' \cup \LR{c}, k, m-1)$ \label{lin:recursivecall}
		\EndFor
		\State Call \texttt{HybridClustering}$(\cen', k, m-1)$ \label{lin:emptycall}
		\State \Return solution $\tilde{\cen}$ minimizing $\cost_{r'}(\cli, \tilde{\cen})$ over recursive calls made in lines \ref{lin:recursivecall} and \ref{lin:emptycall} \label{lin:finalreturn}
	\end{algorithmic}
\end{algorithm}

\Cref{algo:recursive} is a recursive algorithm, and is called {\tt HybridClustering}. It takes three parameters $\cen', k$, and $m$. $\cen' \subset \real^d$ is a subset of centers added to the solution so far and has size $k-m$. Further, $k$ is the total size of the solution, and $m$ is an upper bound on the remaining solution (since we have already added $k-m$ centers). At a high (and imprecise) level, the goal of each recursive step is to find an \emph{approximate replacement} for each center in an unknown optimal solution.

In line~\ref{lin:basecase}, we check whether $m = 0$, i.e., whether we have used our budget of $k$ centers, and if so, we return the same set $F'$ of centers built through the recursive process. Otherwise (line \ref{lin:centergrid} onward), we assume that $m > 0$, i.e., we are yet to add a set of centers. Throughout this process (line \ref{lin:centergrid} to \ref{lin:endouterfor}, we will build a set $R$ consisting of candidate centers, at least one of which will be an approximate replacement of an \emph{unseen} center (i.e., one whose approximate replacement has not already been found) from an optimal solution. Finally, in line \ref{lin:recursivecall}, we will make a recursive call by adding each candidate to the current solution $F'$. Now we discuss how we build the set $R$. 

First, in line \ref{lin:centergrid}, for each center $c' \in F'$ added so far, we add a set of ``nearby'' centers by placing a grid. This handles the case when an unseen optimal center is close to one of the already chosen centers in $F'$. Next, in the outer \textbf{for} loop (line \ref{lin:forloop} to \ref{lin:endouterfor}), we handle the case when all new optimal centers are relatively far from the already chosen centers. In this \textbf{for} loop, we iterate over a range of values for the parameter $q$ via exponential search. Parameter $q$ tries to approximate half of the minimum distance between the already chosen and new optimal centers. Thus, for the ``correct'' value of $q$, the set of points $C_q$ lying ``far'' from the centers of $F'$ (line \ref{lin:faraway}), leaves all of the $m$ unseen optimal clusters untouched. At this point, we aim to use a sample of faraway size (chosen in line \ref{lin:sq-sample}), to find an approximate replacement for one of these $m$ unseen centers. We do this by using the sample in two different ways, to handle two different situations. First, if our sample happens to contain a point ``nearby'' an unseen center, say $c^\star$, then the points chosen from the fine grid in line \ref{lin:samplegrid} will find such an approximate replacement for $c^\star$. Otherwise, the idea is that, if we have removed a significant fraction of points from the ``seen'' clusters in line \ref{lin:faraway}, by virtue of being close to $F'$, then the sample contains sufficiently many (i.e., at least $\beta$) points from the largest unseen cluster, say $C^\star$, with reasonable probability, and these points can be used to find an approximate replacement of the cluster center (using \Cref{prop:samplemedian}). However, \emph{a priori} we do not know which subset of the sample comes from $C^\star$. Therefore, we iterate over all subsets of size $\beta$ in the inner \textbf{for} loop (lines \ref{lin:innerforloop} to \ref{lin:endinnerfor}) to find such a subset of size $\beta$ that comes entirely from $C^\star$ and use a known subroutine, called $\texttt{ApproxSolutionOnSample}$, to find an approximate replacement. Finally, in \ref{lin:recursivecall}, we make a recursive call by adding each center from $R \setminus F'$, and in line \ref{lin:emptycall}, we make a recursive call by not adding any new center (to handle a particular case). In line \ref{lin:finalreturn}, we return the minimum-cost solution found over all recursive calls. This completes the description of the algorithm.

\subsection{Analysis}
The crux of the analysis is to establish that \Cref{algo:recursive} satisfies the following invariant.

\paragraph*{Invariant} Let $0 \le m \le k$ and $0 < \alpha < 1$ be a constant. Suppose for the given $\cen'$ of size at most $k-m$, there exists some $\cen = \cen' \uplus \cen_{o} \subset \real^d$, such that 
\begin{enumerate}
	\item $|\cen_o| \le m$, and
	\item 
	\begin{equation}
		\sum_{c \in \cen'} 	\cost_{r'}(\cluster(c, \cen), c) + \sum_{c \in \cen_o}  \cost_r(\cluster(c, \cen), c) \le (1+ \delta)^{k-m}\cdot \OPT_r \label{eqn:invariant} 
	\end{equation}
\end{enumerate}
Then, with probability at least $\alpha^{m}$, the algorithm returns a solution $\tilde{\cen} \subset \real^d$, such that
\begin{enumerate}
	\item $|\tilde{\cen}| \le k$,
	\item $\cen' \subseteq \tilde{\cen}$, and
	\item \begin{equation}
		\sum_{c \in \tilde{\cen}}	\cost_{r'}(\cluster(c, \cen), c) \le (1+\delta)^{k} \cdot \OPT_r. \label{eqn:invariant2}
	\end{equation}
\end{enumerate} 

\paragraph*{Proof of Correctness.} The proof is by induction on $m$. For the base case, consider $m = 0$. In the base case (\Cref{lin:basecase}), we return the same $\cen' = \cen$ with probability one. In this case, the invariant tells us that $\cost_{r'}(\cli, \cen) \le (1+\delta)^k \cdot \opt_r$, which is what we need to prove. Now we assume that the claim is true for some $m-1 \ge 0$ and we prove it for $m$ by considering different cases.

\subparagraph{Easy case: $\cen = \cen'$.} This is a much simpler case since we have already found the desired set. In this case, any solution $\tilde{\cen}$ returned by a recursive call always contains $\cen = \cen'$ as a subset. Then, in this case, we have that:
 	\begin{align*}
		\cost_{r'}(\cli, \tilde{\cen}) \le \cost_{r'}(\cli, \cen) \le (1+ \delta)^{k-m}\cdot \OPT_r \le (1+\delta)^k \cdot \OPT_r
	\end{align*}
	Here, the first inequality follows from the assumption that $\tilde{\cen} \supseteq \cen = \cen'$, and the second inequality follows from  (\ref{eqn:invariant}) of the invariant. Note that we do not need to rely on the induction here.

\subparagraph{Main case: $\cen' \subsetneq \cen$.} This is the case where we are yet to discover some subset (namely, $\cen \setminus \cen'$) of centers. We will analyze this case by considering different scenarios based on the inter-center distances, as well as their relative sizes. 
	
	First, since $\cen' \subsetneq \cen$, there exists some $c \in \cen \setminus \cen'$. Now, let $c \in \cen \setminus \cen'$ and $c' \in \cen'$ be the pair of centers with the smallest distance, i.e., $(c, c')$ is a pair realizing $\min_{c_1 \in \cen \setminus \cen', c_2 \in \cen'} \d(c_1, c_2)$. Now we consider different cases depending on $\d(c, c')$, namely the closest distance between an already chosen center $c' \in \cen'$, and an ``unseen center'' $c \in \cen\setminus \cen'$.

\subparagraph{Case 1. Nearby center: $\bm{\d(c, c') \le 16r$}.} In this case, via \Cref{obs:grid}, we conclude that there exists some $\tilde{c} \in \grid(c', 16r, \delta r)$ with $\d(\tilde{c}, c) \le \delta r$. Let $\tilde{\cen} = \cen' \cup \LR{\tilde{c}}$. Then, the proof follows from the following claim (see \Cref{fig:nearbygrid}). 

\begin{figure}[t]
	\centering
	\includegraphics[scale=0.6,page=2]{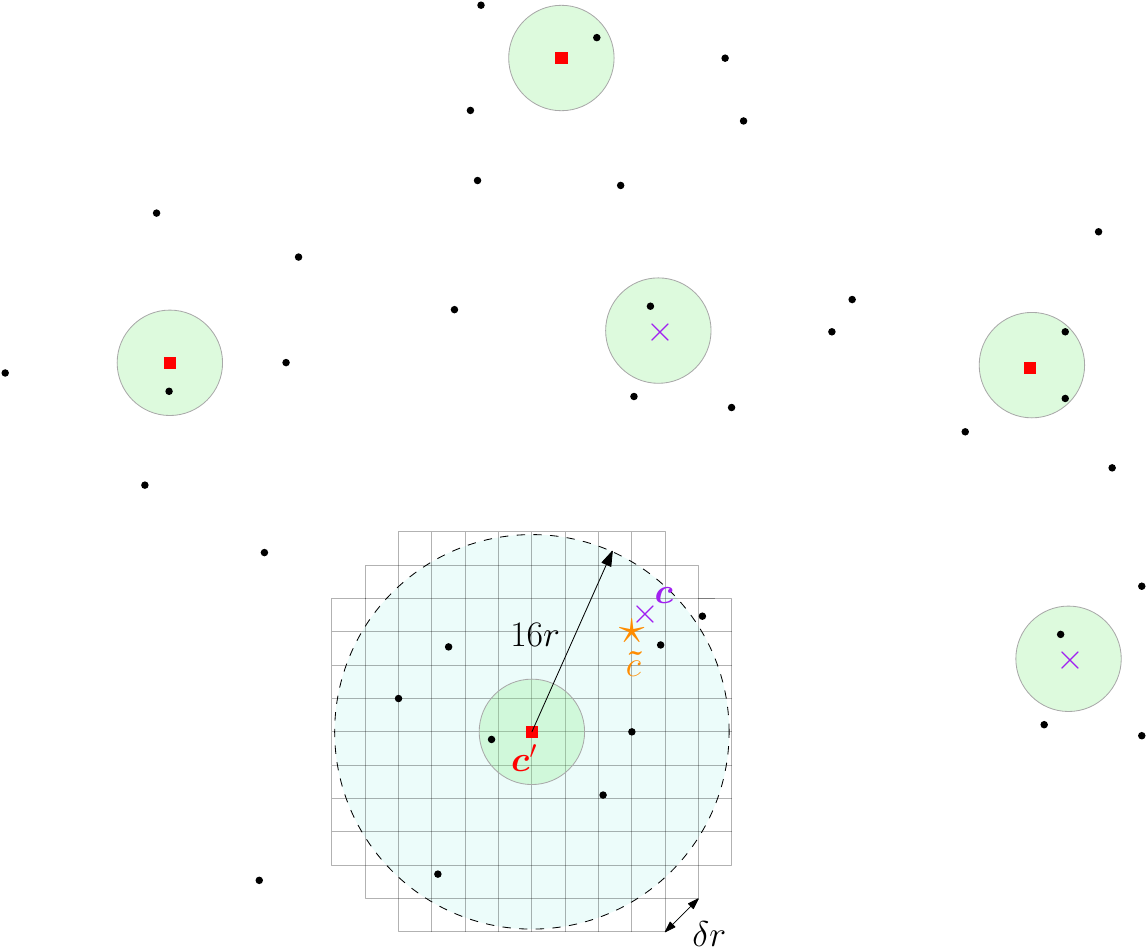}
	\caption{\small Illustration for Case 1. Centers in $\cen'$ are shown as red squares and unseen centers of $\cen \setminus \cen'$ are shown as purple crosses. {\color{Plum}$c$} is the closest center to \red{$F'$} and $\d({\color{Plum}c},\red{c'}) \le 16r$. Then, a nearby center {\color{orange}$\tilde{c}'$} can be found using a $\delta r$ grid.} \label{fig:nearbygrid}
\end{figure}

\begin{claim}\label{cl:nearbycenter}
	Let $c_1 \in \cen_o$ and let $\tilde{c_1} \in \real^d$ be such that $\d(c_1, \tilde{c_1}) \le \delta' r$. Then, with probability at least $\alpha^{m-1}$, {\tt HybridClustering}$(\cen \cup \tilde{c_1}, k, m-1)$ returns a solution $\tilde{\cen}$ that satisfies the required properties.
\end{claim}

\begin{proof}
	Consider $c_1, \tilde{c_1}$ as defined in the statement.
	Let $A = \cluster(c_1, \cen)$. For any point $p \in A$, $\d(p, \tilde{c_1}) \le \d(p, c_1) + \d(c_1, \tilde{c_1}) \le \d(p, c_1) + \delta' r$.
	This implies that, $\d_{r'}(p, \tilde{c_1}) \le \d_{r'}(p, c_1)$.	Define $\cen_{\new} \coloneqq \cen'_{\new} \uplus \cen'_o$, where $\cen'_{\new} \coloneqq \cen' \cup \LR{\tilde{c_1}}$ and $\cen'_o \coloneqq \cen_o \setminus \LR{c_1}$. 
	First, we show the following inequality.
	\begin{align}
		\sum_{c \in \cen'_{\new}} \cost_{r'}(\cluster(c, \cen), c) + &\sum_{c \in \cen'_o} \cost_{r}(\cluster(c, \cen_{\new}), c)\nonumber\\
		 \le& \sum_{c \in \cen'}  \cost_{r'}(\cluster(c, \cen),c) + \sum_{c \in \cen_o} \cost_{r}(\cluster(c, \cen), c) \label{eq:claimbound}
	\end{align}
	We construct an assignment of clients to the centers in $\cen_{\new}$, where we may not assign a client to its closest center. To construct this assignment, we consider different cases. For $c \in \cen' \cup \cen_o \setminus \LR{c_1}$, we assign all points $p \in \cluster(c, \cen)$ to $c$. The contribution of all such points is the same as the right-hand side of (\ref{eq:claimbound}). Finally, we assign all points in $\cluster(c_1, \cen)$ to $\tilde{c_1}$. By the choice of $\tilde{c}_1$, $\cost_{r'}(\cluster(c_1, \cen), \tilde{c_1}) \le \cost_{r}(\cluster(c_1, \cen), c_1)$, which is the contribution of such points on the right-hand side. Since the cost on the left-hand side is no larger than the cost of the assignment thus constructed, it shows (\ref{eq:claimbound}). 
	
	Note that the right-hand side of (\ref{eq:claimbound}) is at most $(1+\delta)^{k-m} \cdot \opt_r$ due to the invariant, and hence $\cen_{\new}$ satisfies the properties required to apply the inductive hypothesis for $m-1$. This implies that with probability at least $\alpha^{m-1}$ the recursive call \texttt{Recursive}$(\cen'\cup \LR{\tilde{c_1}}, k, m-1)$ returns a solution $\tilde{\cen}$ satisfying $\cost_{r'}(\cli, \tilde{\cen}) \le (1+\delta)^k \cdot \OPT_r$.
\end{proof}

\subparagraph{Case 2. Faraway center: $\bm{16r < \d(c, c') \le \d_{\max}}$}. Let $t = \d(c, c')$ and $q^\star$ be the largest power of $2$ that is at most $t/2$.  Consider $\cli_{q^\star} = \cli \setminus \lr{ \bigcup_{c_1 \in \cen'} B(c_1, q^\star)}$. Let $c^\star \in \cen \setminus \cen'$ denote the center of the maximum-size cluster, i.e., $c^\star = \arg\max_{c_1 \in \cen \setminus \cen'} |\cluster(c_1, \cen)|$, and $L \coloneqq \cluster(c^\star, \cen)$ denote the largest cluster. Finally, let $D \coloneqq \bigcup_{c_{\text{old}} \in \cen'} \cluster(c_{\text{old}}, \cen) \cap \cli_{q^\star}$ denote the set of clients that are distant from the respective centers in $\cen'$. Let us summarize some consequences of these definitions in the following observation (its proof is essentially discussed above). Also see \Cref{fig:farcluster}.

\begin{figure}[t]
	\centering
	\includegraphics[scale=0.5,page=3]{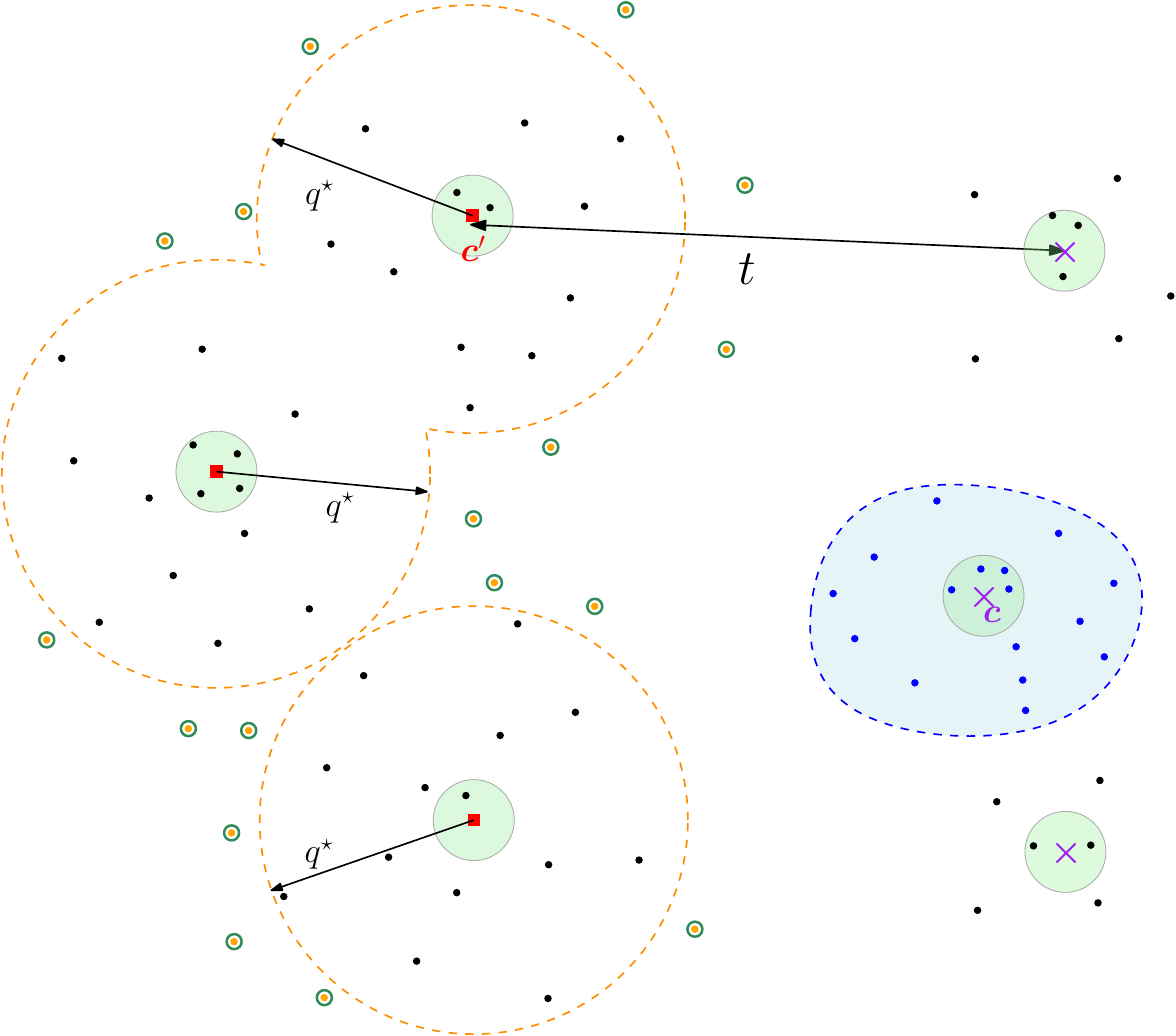}
	\caption{\small Illustration for Case 2. Centers of $\red{\cen'}$ are shown as red squares and unseen centers of ${\color{Plum}\cen \setminus \cen'}$ are shown as purple crosses. Balls of radius $q^\star$ around $\cen'$ are shown in dashed orange. $P'$ are the points lying outside these balls. Among the points of $P'$, $D$ is the set of points belonging to clusters around $F'$, and shown as green-orange filled dots. Finally, the cluster around ${\color{Plum} c}$ is the largest unseen cluster (marked in dashed blue shape), $L$. We analyze different cases depending on the relative sizes of $L$ and $D$.} \label{fig:farcluster}
\end{figure}

\begin{observation} \label{obs:definitions}~
	\begin{enumerate}
		\item $\displaystyle \cli_{q_\star} = D \uplus \biguplus_{c_1 \in \cen \setminus \cen'}\cluster(c_1, \cen) \cap \cli_{q^\star}$. 
		\item In particular, $\cluster(c, \cen), \cluster(c^\star, \cen) \subseteq \cli_{q^\star}$.
	\end{enumerate}
\end{observation}
We consider different sub-cases based on the relative sizes of $L$ and $D$.

\subparagraph{Case 2.1. New cluster is tiny: $\bm{|L| \le \delta^2/4 \cdot |D|}$.} Let $N \coloneqq \cluster(c, \cen)$. Note that the definition of $c^\star$, combined with the case assumption, implies that $|N| \le |L| \le \frac{\delta^2}{4} \cdot |D|$. We summarize a few technical consequences of these definitions in the following claim. 

\begin{claim}\label{cl:tinycluster}
	\begin{align}
		\cost_r(N, c') &\le \delta \cdot \cost_r(D, \cen) + (1+\delta) \cdot \cost_r(N, c) \label{eqn:costincrease}
	\end{align}
\end{claim}

\begin{proof}
	Note that for each $p \in D$, $\d(p, \cen) \ge \frac{t}{4}$. Thus, $\d_r(p, c') \ge \frac{t}{4} - r \ge \frac{3t}{16}$, where the last inequality follows from the case assumption, namely $t > 16r$. Thus, each point $p \in D$ contributes at least $\frac{3t}{16}$ to $\cost_{r'}(\cli, \cen)$, and their total contribution to $\cost_r(\cli, \cen)$ is 
	\begin{equation}
		\cost_r(D, \cen) \coloneqq \sum_{p \in D} \d_r(p, c') \ge |D| \cdot \frac{3t}{16} \label{eqn:costd}
	\end{equation}
	
	Now we upper bound the cost of assigning points of $N$ to $c'$. To this end, we partition $N = \Nnear \uplus \Nfar$, where $\Nnear \coloneqq \LR{p \in A: \d(p, c) \le 2r/\delta}$ and $\Nfar \coloneqq \LR{p \in A: \d(p, c) > 2r/\delta}$. Note that, for each $p \in \Nfar$, $\d_r(p, c) = \d(p, c) - r \ge \d(p, c) - \tfrac{\delta}{2} \cdot \d(p, c)$, which implies that, for each $p \in \Nfar$, 
	\begin{equation}
		\d(p, c) \le \lr{\frac{1}{1-\tfrac{\delta}{2}}} \cdot \d_r(p, c) \le (1+\delta) \cdot \d_r(p, c) \label{eqn:nfar}
	\end{equation} 
	Now consider,
	\begin{align}
		\cost_r(N, c') &\le \sum_{p \in N} \d(p, c') \tag{Since $\d_r(\cdot, \cdot) \le \d(\cdot, \cdot)$}
		\\&\le \sum_{p \in N} \d(p, c) + \d(c, c') \tag{Triangle inequality}
		\\&= |N| \cdot \d(c, c') + \sum_{p \in \Nnear} \d(p, c) + \sum_{p \in \Nfar} \d(p, c)  \nonumber
		\\&\le \tfrac{\delta^2}{8} \cdot |D| \cdot t + \sum_{p \in \Nnear} \tfrac{4r}{\delta} + \sum_{p \in \Nfar} (1+\delta) \cdot \d_r(p, c) \tag{From case assumption and (\ref{eqn:nfar})}
		\\&= \tfrac{\delta^2}{4} \cdot |D| \cdot t + \tfrac{4r}{\delta} \cdot |\Nnear| + (1+\delta) \cdot \cost_r(N, c) \nonumber
		\\&\le \tfrac{\delta^2}{4} \cdot |D| \cdot t + \delta \cdot |D| \cdot \tfrac{t}{16} + (1+\delta) \cdot \cost_r(N, c) \tag{$|\Nnear| \le |N| \le \tfrac{\delta^2}{4}$ and $t > 16r$}
		\\&\le |D| \cdot \tfrac{3t}{16} \cdot \delta \cdot \lr{ \tfrac{4\delta}{3} + \tfrac{1}{3}} + (1+\delta)\cdot \cost_r(N,c) \nonumber
		\\&\le \delta \cdot \cost_r(D, \cen) + (1+\delta) \cdot \cost_r(N, c) 
	\end{align}
	Where the last inequality follows from (\ref{eqn:costd}) and $\delta < 1/2$. 
\end{proof}

Thus, consider the solution $\cen \setminus \LR{c}$. To upper bound $\cost_r(\cli, \cen \setminus \LR{c})$, we assign all points in $N = \cluster(c, \cen)$ to $c'$. The cost of this solution can be upper bounded as follows
\begin{align}
	\cost_r(\cli, \cen \setminus \LR{c}) &\le \cost_r(\cli, \cen) - \cost_r(N, c) + \cost_r(N, c') \nonumber
	\\&\le \cost_r(\cli, \cen) + \delta \cdot \cost_r(D, \cen) + \delta \cdot \cost_r(N, c) \tag{From (\ref{eqn:costincrease})}
	\\&\le \cost_r(\cli, \cen) + \delta \cdot \cost_r(\cli, \cen) \tag{Since $D \uplus N \subseteq \cli$}
	\\&\le (1+\delta)^{k-m+1} \cdot \OPT_r \label{eqn:costbound}
\end{align}
Where the last inequality follows from the invariant. Further, observe that $c \in \cen \setminus \cen'$, which implies that $|(\cen \setminus \LR{c}) \setminus \cen'| \le m-1$. This, combined with (\ref{eqn:costbound}), shows that the solution $\cen \setminus \LR{c} = \cen' \uplus (\cen_o \setminus \LR{c})$ satisfies the conditions of the invariant for $m-1$. Then, by using inductive hypothesis, $\texttt{HybridClustering}(\cen, k, m-1)$, with probability at least $\alpha^{m-1} \ge \alpha^m$, returns a solution $\tilde{\cen}$ such that (a) $|\tilde{\cen}| \le k$, (b) $\cen \subseteq \tilde{\cen}$, and (c) $\cost_{(1+\delta)r'}(\cen, \tilde{\cen}) \le (1+\delta)^{k} \cdot \OPT_r$, completing the induction.

\subparagraph{Case 2.2. New cluster is large enough: $\bm{|L| > \delta^2/4 \cdot |D|}$.} That is, the largest ``untouched'' cluster is at least an $\delta$ fraction of the remaining points. Since $L = \cluster(c^\star, \cen)$ is the largest ``untouched'' cluster, $|L| \ge |\cluster(c_1, \cen)|$ for all $c_1 \in \cen \setminus \cen'$. Then, by \Cref{obs:definitions}, we have that,
\begin{align}
	|\cli_{q^\star}| = |D| + \sum_{c_1 \in \cen \setminus \cen'} |\cluster(c_1, \cen)| \le |D| + |\cen \setminus \cen'| \cdot |L| &\le \tfrac{4}{\delta^2} \cdot |L| + k \cdot |L| \le \tfrac{5k}{\delta^2} \cdot |L| \nonumber
\end{align}
In other words, $|L| \ge \frac{\delta^2}{5k} \cdot |\cli_{q^\star}|$. 

In the next claim, we summarize some properties of the sample $S$ chosen in line 7 of the algorithm, in the current case, i.e., when $|L| \ge \frac{\delta^2}{5k} |C_{q^\star}|$.
\begin{claim} \label{prop:probbound}
	Consider the iteration of the \textbf{for} loop of \Cref{lin:forloop}, when $q = q^\star$, and the corresponding sample $S_q$ obtained in \Cref{lin:sq-sample}. The following statements hold. 
	\begin{enumerate}
		\item With probability at least $1/2$, $S_q$ contains at least $\beta$ points of $L$.
		\item $S_q \cap L$ has the same distribution as selecting $|S_q \cap L|$ points uniformly at random from $L$.
		\item Let $L' \subseteq L$ be an arbitrary subset of size at least $\frac{\delta}{10} |L|$. Then, with probability at least $1/2$, $S_q$ contains at least $1$ point from $L$, i.e., $S_q \cap L' \neq \emptyset$.
	\end{enumerate}
\end{claim}
\begin{proof}
	Recall that $\beta' = \beta \cdot \tfrac{150k}{\delta^3}$ and $|L| \ge \tfrac{\delta^2}{5k}|C_{q^\star}|$. So, the first item follows, say, via Markov's inequality \footnote{In fact, a closer inspection reveals that the probability is much closer to $1$, but ``at least $1/2$'' suffices for our purpose.}. The second item is an easy consequence of conditional distributions. The proof of the third item is analogous to the first item, combined with the bound on $|L'|$.
\end{proof}

Now we condition on the event that $|S_{q^\star} \cap L| \ge \beta$, which, by \Cref{prop:probbound} happens with probability at least $\frac{1}{2}$. Then, let $S' \subseteq S \cap L$ be such a subset of size $\beta$. Let $L = \Lnear \uplus \Lfar$, where $\Lnear = \LR{p \in L: \d(p, c^\star) \le \tfrac{8r}{\delta}}$ and $\Lfar = \LR{p \in L: \d(p, c^\star) > \tfrac{8r}{\delta}}$. We consider different cases depending on the relative sizes of $\Lnear$ and $\Lfar$. In the first case below (2.2.1), when $|\Lnear|$ is not very tiny compared to $|\Lfar|$, we show that our sample contains at least one point from $\Lnear$ with good probability, and hence an $\varepsilon r$ grid around that point will contain an approximate center. In the complementary case (2.2.2), $|\Lnear|$ is very tiny compared to $|\Lfar|$, and in this case, we argue that, instead of finding an approximate \textsc{Hybrid $1$-Median}, we can focus on finding an approximate \textsc{$1$-Median}, which can be found using the sample. Now we formally analyze each of these cases.

\subparagraph{Case 2.2.1. $\bm{|\Lnear| > \frac{\delta}{8} \cdot |\Lfar|}$.} 
In this case, letting $L' \gets \Lnear$ in \Cref{prop:probbound}, we infer that with at least probability $1/2$, $S' \cap \Lnear \neq \emptyset$. We condition on this event. Then, since $\d(p, c^\star) \le \frac{8r}{\delta}$, it follows that there exists a $\tilde{c}^\star \in \grid(p, \frac{8r}{\delta}, \delta r)$, such that $\d(\tilde{c}^\star, c^\star) \le \delta r$. Since we branch on each point in $\bigcup_{p' \in S} \grid(p, \frac{8r}{\delta}, \delta r)$, we will branch on $\tilde{c}^\star$ in particular. Then, by \Cref{cl:nearbycenter}, \texttt{HybridClustering}$(\cen' \cup \LR{\tilde{c}^\star}, k, m-1)$ returns a solution $\tilde{\cen}$, with probability at least $1/2 \cdot \alpha^{m-1} \ge \alpha^{m}$.

\subparagraph{Case 2.2.2. $\bm{|\Lnear| \le \frac{\delta}{8} \cdot |\Lfar|}$.}

We prove two claims, namely \Cref{cl:firstclaim}, and \Cref{cl:approx-hybrid-equiv}. The latter essentially reduces the problem to finding an approximate solution to $1$-median on $L$. Intuitively speaking, this follows from the following two reasons: (1) As we show in (\ref{eqn:lfar-approx-bound}), a similar statement holds for the points in $\Lfar$. This essentially follows from the fact that, since each point of $\Lfar$ has distance at least $\frac{8r}{\delta}$ to $c$, the subtraction of $r$ from their distances has little effect on the cost, and (2) Due to the case assumption, the points of $\Lfar$ vastly outnumber the points of $\Lnear$. Hence, the preceding claim also translates to the points of $L = \Lfar \cup \Lnear$, at a further small approximation error.

\begin{claim}\label{cl:firstclaim}\ \  $\displaystyle
		\sum_{p \in L} \d_r(p, c^\star) \le \sum_{p \in L} \d(p, c^\star) \le (1+\tfrac{3\delta}{4}) \cdot \sum_{p \in L} \d_r(p, c^\star)$.
\end{claim}

\begin{proof}
	First, consider,
	\begin{align}
		\sum_{p \in \Lfar} \d_{r}(p, c^\star) = \sum_{p \in \Lfar} \d(p, c^\star) - r &\ge \sum_{p \in \Lfar} \d(p, c^\star) - \tfrac{\delta}{8} \cdot \d(p, c^\star)\tag{Since $\d(p, c^\star) \ge \frac{8r}{\delta} > r$} \label{eqn:lfar-part1}
	\end{align}
	Then, the inequality between the first and the last term can be rewritten as,
	\begin{equation}
		\sum_{p \in \Lfar}\d(p, c^\star) \le \frac{1}{1-\delta/8} \cdot \sum_{p \in \Lfar} \d_r(p, c^\star) \le (1+\tfrac{\delta}{4}) \cdot \sum_{p \in \Lfar} \d_r(p, c^\star) \label{eqn:lfar-approx-bound}
	\end{equation}
	
	The following inequality will be used later to show that the contribution of points of $\Lnear$ is negligible to the overall cost.
	\begin{align}
		\sum_{p \in \Lfar} \d(p, c^\star) \ge \sum_{p \in \Lfar} \d_r(p, c^\star) &\ge \sum_{p \in \Lfar}  (1-\tfrac{\delta}{8}) \cdot \d(p, c^\star) \tag{From (\ref{eqn:lfar-approx-bound})}
		\\&\ge (1-\tfrac{\delta}{8}) \cdot \tfrac{8r}{\delta} \cdot |\Lfar| \tag{Definition of $\Lfar$}
		\\&\ge \tfrac{1}{2} \cdot \tfrac{8r}{\delta} \cdot \tfrac{8}{\delta} \cdot |\Lnear| \tag{Case assumption: $|\Lfar| \ge \frac{8}{\delta} \cdot |\Lfar|$}
		\\&= \tfrac{4}{\delta} \cdot \tfrac{8r}{\delta} \cdot |\Lnear| \nonumber
		\\&\ge \tfrac{4}{\delta} \cdot \sum_{p \in \Lnear} \d_r(p, c^\star) \label{eqn:lnear-bound}
	\end{align}
	Where the last inequality follows from the definition of $\Lnear$.
	
	The next sequence of inequalities shows a bound similar to (\ref{eqn:lfar-approx-bound}), but when the sum is taken over all points of $L$ (instead of only the points of $\Lfar$, as in \Cref{eqn:lfar-approx-bound}).
	\begin{align}
		\sum_{p \in L} \d_r(p, c^\star) \le \sum_{p \in L} \d(p, c^\star) &= \sum_{p \in \Lnear} \d_r(p, c^\star) + \sum_{p \in \Lfar} \d_r(p, c^\star) \nonumber
		\\&= (1+\tfrac{\delta}{4}) \cdot \sum_{p \in \Lfar} \d(p, c^\star) \tag{From (\ref{eqn:lnear-bound})}
		\\&\le (1+\tfrac{\delta}{4}) \cdot (1+\tfrac{\delta}{4}) \cdot \sum_{p \in \Lfar} \d_r(p, c^\star) \tag{From (\ref{eqn:lfar-approx-bound})}
		\\&\le (1+\tfrac{3\delta}{4}) \cdot \sum_{p \in L} \d_r(p, c^\star) 
	\end{align}
	Where the last inequality follows from (i) $(1+\tfrac{\delta}{4}) \cdot (1+\tfrac{\delta}{4}) \le 1+\tfrac{3\delta}{4}$ and (ii) $\Lfar \subseteq L$. This completes the proof of the claim.
\end{proof}

Using this claim, we prove the following claim, which shows that it is sufficient to find an approximate $1$-median solution for $L$, which will also be a good approximation for \textsc{Hybrid $1$-Median} for $L$. To this end, let $\tilde{c}^\star \in \real^d$ denote the optimal $1$-median for $L$. 
\begin{claim} \label{cl:approx-hybrid-equiv}
	Let $c_1$ be an $(1+\tfrac{\delta}{8})$-approximation for $1$-\textsc{Median} for $L$, i.e., $\sum_{p \in L} \d(p, c_1) \le (1+\tfrac{\delta}{8}) \cdot \sum_{p \in L} \d(p, \tilde{c}^\star)$. Then, it is also a $(1+\delta)$-approximation for \textsc{Hybrid $1$-Median} for $L$, i.e., 
	\begin{equation}
		\sum_{p \in L} \d_r(p, c_1) \le (1+\delta) \cdot \sum_{p \in L} \d_r(p, c^\star) \label{eqn:approx-hybrid-equiv}
	\end{equation}
\end{claim}
\begin{proof}
	Let $c_1, \tilde{c}^\star \in \real^d$ as defined above. Then, 
	\begin{align*}
		\sum_{p \in L} \d_r(p, c_1) \le \sum_{p \in L} \d(p, c_1) &\le (1+\tfrac{\delta}{8}) \cdot \sum_{p \in L} \d(p, \tilde{c}^\star) \tag{By definition of $c_1$}
		\\&\le (1+\tfrac{\delta}{8}) \cdot \sum_{p \in L} \d(p, c^\star) \tag{Since $c_1 \in \real^d$ is an optimal $1$-median and $c^\star \in \real^d$ is a feasible median}
		\\&\le (1+\tfrac{\delta}{8}) \cdot (1+\tfrac{3\delta}{4}) \cdot \sum_{p \in L} \d_r(p, c^\star) \tag{From \Cref{cl:firstclaim}}
		\\&\le (1+\delta) \cdot \sum_{p \in L} \d_r(p, c^\star) \qedhere
	\end{align*}
\end{proof}
Thus, now the task reduces to finding a $1+\tfrac{\delta}{8}$-approximate $1$-\textsc{Median} solution for $L$. To this end, we have the following result from \cite{KumarSS05,KumarSS10}.

\begin{proposition}[\cite{KumarSS05,KumarSS10}] \label{prop:samplemedian}
	Let $X \subset \real^d$ be a set of $n$ points and $0 < \delta < 1$. Let $S \subseteq X$ be a uniform sample chosen from $X$ of size $\beta = \lr{\frac{1}{\delta}}^{c'}$. Then, there exists an algorithm that runs in time $2^{\Oh(1/\delta^c)} d$, and with probability at least $\alpha'$, returns an $(1+\delta)$-approximate $1$-median for $X$. Here, $c, c'$ are absolute constants independent of the dimension $d$.
\end{proposition}
We combine the properties of the sample $S_{q^\star}$ proved in \Cref{prop:probbound} along with the previous proposition, to complete the proof. To this end, note that the first item of \Cref{prop:probbound} implies that, with probability at least $1/2$, $S_{q^\star}$ contains at least $\beta = \frac{1}{\delta^{c}}$ points from $L$. Then, we use the algorithm of \Cref{prop:samplemedian}, that returns with probability at least $\alpha'$, a $(1+\frac{\delta}{8})$-approximate $1$-median $c_1 \in \real^d$ for $L$. It follows that,
\begin{align}
  &\sum_{c \in \cen' \cup \LR{c_1}}\cost_{r'}(\cluster(c, \cen), c) + \sum_{c \in \cen_o \setminus \LR{c^\star}}\cost_{r}(\cluster(c, \cen), c) \nonumber\\
	\le& \sum_{c \in \cen'}\cost_{r'}(\cluster(c, \cen), c) +  \sum_{c \in \cen_o}\cost_{r}(\cluster(c, \cen), c) - \cost_r(\cluster(c^\star, \cen), c^\star) + \sum_{p \in L} \d_r(p, c_1) \nonumber\\
	\le & \sum_{c \in \cen'}\cost_{r'}(\cluster(c, \cen), c) +  \sum_{c \in \cen_o}\cost_{r}(\cluster(c, \cen), c) - \cost_r(\cluster(c^\star, \cen), c^\star)\nonumber \\
	& + \delta \cdot \sum_{c \in \cen_o}\cost_{r}(\cluster(c, \cen), c) \tag{From \Cref{cl:approx-hybrid-equiv}}\\
	\le &(1+\delta) \cdot \lr{\sum_{c \in \cen'}\cost_{r'}(\cluster(c, \cen), c) + \sum_{c \in \cen_o}\cost_{r}(\cluster(c, \cen), c)} \nonumber\\
	\le & (1+\delta)^{k-m+1} \cdot \OPT_r
\end{align}
Then, by induction hypothesis, \texttt{HybridClustering}$(\cen' \cup \LR{c_1}, k, m-1)$, with probability at least $\alpha^{m-1}$, returns a solution $\tilde{\cen}$ such that $\cost_{r'}(\cli, \tilde{\cen}) \le (1+\delta)^k \cdot \OPT_r$. The overall probability of this event is at least $\frac{1}{2} \cdot \alpha' \cdot \alpha^{m-1} = \alpha^{m}$, completing the induction.

This finishes the case analysis, and thus we have established the invariant using induction. Using the invariant, we can show the following key lemma.

\begin{lemma} \label{lem:runtime}
	{\tt HybridClustering}$(\emptyset, k, k)$ returns a $(1+\varepsilon, 1+\varepsilon)$-bicriteria approximation solution to the given instance of \probname with probability at least $\alpha^k$ for some constant $0 < \alpha < 1$. 
\end{lemma}
\begin{proof}
	We first show the following:
	\begin{claim} $\displaystyle |R| \le \lr{\tfrac{k\sqrt{d}}{\delta}}^{\Oh(d)} + \frac{k \log n}{\delta^{\Oh(1)}} \cdot \lr{\lr{\tfrac{\sqrt{d}}{\delta}}^{\Oh(d)} + \binom{\beta'}{\beta}} \le (\log n) \cdot 2^{(kd/\delta)^{\Oh(1)}}.$
	\end{claim}
	\begin{proof}
		First, in \cref{lin:centergrid}, we add the points returned by $\grid(c', 16r, \delta r)$ for each $c' \in \cen'$. The number of such points is $\lr{\frac{16r\sqrt{d}}{\delta r}}^d = \lr{\frac{\sqrt{d}}{\delta}}^{\Oh(d)}$. Next, there are at most $\frac{\log n}{\delta^{\Oh(1)}}$ values for $q$ (this follows from the second preprocessing step, cf.~\Cref{lem:aspectratio}), corresponding to each iteration of the for loop. In each iteration, we take a sample $S$ of size $\beta' = \Oh{\lr{\frac{150k\beta}{\delta^3}}}$. Then, for each $p \in S$, we add to $R$ the points of $\grid(p, 8r/\delta, \delta r)$, and the number of such points is at most $\lr{\frac{8\sqrt{d}}{\delta^2}}^d = \lr{\frac{\sqrt{d}}{\delta^2}}^{\Oh(d)}$. In addition, we iterate over each subset $S' \subseteq S$ of size $\beta$, and the number of such subsets is $\binom{\beta'}{\beta} \le \lr{\frac{e \beta'}{\beta}}^\beta = \lr{\frac{k}{\delta^3 \cdot \beta}}^{\beta} = \lr{\frac{k}{\delta^3}}^{(1/\delta)^{\Oh(1)}} = k^{1/\delta^{\Oh(1)}}$. Thus, overall, the size of $R$ is bounded by $(\log n) \cdot 2^{(kd/\delta)^{\Oh(1)}}$.
	\end{proof}
	To bound the running time of the algorithm, let $T(m)$ denote an upper bound on \\\texttt{HybridClustering}$(F', k, m)$ for any $F' \subset \real^d$. Note that we make a recursive call on each point in $R$. Further by \Cref{prop:samplemedian}, the time taken to compute a center in \cref{lin:approxsample} is at most $2^{(1/\delta)^{\Oh(1)}}$; and this algorithm is used in each of the at most $\frac{k \log n}{(1/\delta)^{\Oh(1)}} \cdot \binom{\beta'}{\beta} \le (\log n) \cdot k^{\lr{1/\delta}^{\Oh(1)}}$. Thus, $T(m)$ can be bounded by the following recurrence.
	\begin{align*}
		T(m) &\le |R| \cdot T(m-1) + (\log n) \cdot k^{\lr{1/\delta}^{\Oh(1)}} \cdot n^{\Oh(1)}
		\\&\le (\log n) \cdot 2^{\lr{\frac{kd}{\delta}}^{\Oh(1)}} \cdot T(m-1) + k^{\lr{1/\delta}^{\Oh(1)}} \cdot n^{\Oh(1)}
	\end{align*}
	It can be shown that this recurrence solves to $T(m) \le 2^{\lr{\frac{kd}{\delta}}^{\Oh(1)}} \cdot n^{\Oh(1)}$ -- here we use the standard argument that $(\log n)^k \le k^{\Oh(k)} \cdot n^{\Oh(1)}$. 
	
	Finally, note that our first call to the recursive algorithm is \texttt{HybridClustering}$(\cen' = \emptyset, k, k)$. At this point, the precondition of the invariant is satisfying setting $\cen_o \gets \cen^*$, an optimal solution satisfying $\cost_r(P, \cen^*) = \opt_r$. Then, the correctness of the invariant implies that, with probability at least $\alpha^k$, the algorithm returns a solution $\tilde{\cen}$ of size at most $k$, that is a $(1+\varepsilon, 1+\varepsilon)$-bicriteria approximation -- here we use that $\delta = \frac{\varepsilon}{10k}$, which implies that $(1+\delta)^k \le (1+\varepsilon)$.
\end{proof}

We now conclude with the following theorem, which is restated for convenience.
\fptasthm*
\begin{proof}
	From \Cref{lem:runtime}, the success probability of the algorithm is $\alpha^k$ for some constant $\alpha > 0$. Thus, we need to repeat the algorithm $\alpha^{-k}$ times to boost the probability to at least a positive constant, which gets absorbed in the $2^{\lr{\frac{kd}{\varepsilon}}^{\Oh(1)}}$ factor. 
\end{proof}

\paragraph*{A Hybrid of {\sc $k$-Center} and $k$-{\sc Means}.} We note that an almost identical algorithm also implies a $(1+\varepsilon, 1+\varepsilon)$ bicriteria approximation for an analogous generalization of $k$-\textsc{Center} and $k$-\textsc{Means}. In this problem, the objective of (\ref{eqn:objective}) is replaced by the following: $\cost_r(C, F) \coloneqq \sum_{p \in C} \d_r(p, C)^2$. Let us refer to this problem as \textsc{Hybrid $(k, 2)$-Clustering} -- the ``$2$'' in the name refers to the squares of the distance-thresholds that feature in the objective. Most of the analysis can be adapted to deal with the squares of the distances, by appropriately changing the sizes and distance-thresholds. The only significant change is that instead of \Cref{prop:samplemedian}, one needs to use an algorithm that computes an approximate $1$-\textsc{Means} solution given a large enough uniform sample of the cluster -- such an algorithm can also be found in \cite{KumarSS10}. Then, one obtains the following theorem.

\begin{theorem}
	Let $0 < \varepsilon < 1$. There exists a randomized algorithm that, given an instance of {\sc Hybrid $(k, 2)$-Clustering} in $\real^d$, runs in time $2^{(\frac{kd}{\varepsilon})^{\Oh(1)}} \cdot n^{\Oh(1)}$ and returns a $(1+\varepsilon, 1+\varepsilon)$-approximation with probability at least a positive constant.
\end{theorem}

More generally, one can also define \textsc{Hybrid $(k, z)$-Clustering} analogously, where the threshold-distances feature the $z$-th power of distances. Again, our approach easily extends to this problem, modulo a version of \Cref{prop:samplemedian} for the vanilla {\sc $(k, z)$-Clustering} in Euclidean spaces. To the best of our knowledge, such an algorithm is not explicitly known in the literature; however, it may be possible to obtain such an algorithm using the approach of \cite{JKS2014,KumarSS10}.

\section{Conclusion and Future Directions} \label{sec:conclusion}

In this paper, we proposed a novel clustering objective and defined a new problem, called \probname, that generalizes both \kmed and  \kcenter. For $d$-dimensional euclidean inputs, we designed a randomized $(1+\varepsilon, 1+\varepsilon)$-bicriteria approximation scheme for \probname running in time $2^{(kd/\varepsilon)^{\Oh(1)}}$, for any $\varepsilon > 0$. Further, essentially the same algorithm also generalizes for a hybrid objective of $k$-\textsc{Center} and $k$-\textsc{Means}. We remind that improving either of the two $(1+\varepsilon)$ factors to $1$ would imply an exact FPT algorithm for \kcenter/\textsc{Median(/Means)} in Euclidean spaces, which is unlikely to exist.

Our work opens up several interesting research directions. An immediate question is whether improving or removing the FPT dependence on the dimension $d$ is possible, similar to the approach in \cite{KumarSS10} for \kmed/\textsc{Means}.
One potential direction for achieving this could be the recent result that imports the famous Johnson-Lindenstrauss dimension reduction technique to $k$-clustering problems \cite{MakarychevMR19}. 
Another intriguing question is the design of coresets for \probname, which could also have some implications for the previous problem via the approach of \cite{Charikar22JL}. However, at a high level, designing coresets for \probname appears to be challenging, since \emph{a priori} we do not know which points belong inside the radius-$r$ balls (and thus contribute $0$ to the cost), and which ones lie outside, and hence their cost needs to be approximately preserved. 

Finally,  considering \probname with inputs from arbitrary metric spaces, a primal-dual algorithm from \cite{ChakrabartyS18} can be adapted to obtain an $(\alpha, \beta)$-bicriteria approximation in polynomial time, for some constants $\alpha$ and $\beta$ \footnote{A quick examination of the proof of \cite{ChakrabartyS18} suggests that $(18, 6)$-bicrtieria approximation easily follows, with further improvements possible with more careful analysis.}. 
Exploring the best possible constants in the bicriteria approximation would be an interesting avenue for future research.


\end{document}